\RequirePackage{amsmath}
\documentclass{llncs}

\newcommand{\solve}{\mathsf{Solve}}
\newcommand{\prune}{\mathsf{Prune}}
\newcommand{\pponly}[1]{}

\usepackage{algorithm}
\usepackage{algpseudocode}
\usepackage{graphicx}
\usepackage{subcaption}
\usepackage{mwe}

\usepackage{siunitx}
\makeatletter
\renewcommand{\Function}[2]{%
        \csname ALG@cmd@\ALG@L @Function\endcsname{#1}{#2}%
        \def\jayden@currentfunction{#1}%
}
\newcommand{\funclabel}[1]{%
        \@bsphack
        \protected@write\@auxout{}{%
                \string\newlabel{#1}{{\jayden@currentfunction}{\thepage}}%
        }%
        \@esphack
}
\makeatother

\newcommand{\lrf}{\mathcal{L}_{\mathbb{R}_{\mathcal{F}}}}
\newcommand{\ucnf}{\text{CNF}^{\forall}}
\newcommand{\norm}[1]{\lVert #1 \rVert}

\usepackage[latin1]{inputenc}
\usepackage[T1]{fontenc}
\usepackage{array,multirow,tabularx}
\usepackage{setspace}
\usepackage{latexsym,stmaryrd,amssymb,mathtools,boxedminipage}
\usepackage{epsfig,graphics,wrapfig,alltt,verbatim}
\usepackage{color}
\usepackage{xspace,url}
\usepackage{paralist}
\usepackage{galois}
\usepackage{listings}
\usepackage[hidelinks]{hyperref}
\hypersetup{colorlinks=false}
\usepackage{todonotes}
\usepackage{stmaryrd}

\newenvironment{spmatrix}
{\begin{small}\begin{pmatrix}}
{\end{pmatrix}\end{small}}

\begin{document}

\title{Delta-Decision Procedures for Exists-Forall Problems over the Reals}
\author{Soonho Kong\inst{1} \and Armando Solar-Lezama\inst{2} \and Sicun Gao\inst{3}}
\institute{
  Toyota Research Institute, Cambridge, USA\\
  \email{soonho.kong@tri.global} \and
  Massachusetts Institute of Technology, USA\\
  \email{asolar@csail.mit.edu} \and
  University of California, San Diego, USA\\
  \email{sicung@ucsd.edu}}

\maketitle
\begin{abstract}
  We propose $\delta$-complete decision procedures for solving satisfiability of nonlinear
  SMT problems over real numbers that contain universal quantification and a wide range
  of nonlinear functions. The methods combine interval constraint propagation, counterexample-guided synthesis, and numerical optimization. In particular, we show how to handle the interleaving of numerical and symbolic computation to ensure delta-completeness in quantified reasoning. We demonstrate that
  the proposed algorithms can handle various challenging global optimization and control synthesis problems that are beyond the
  reach of existing solvers.
\end{abstract}


\section{Introduction}\label{sec:introduction}

Much progress has been made in the framework of delta-decision procedures for solving nonlinear Satisfiability Modulo Theories (SMT) problems over real numbers~\cite{DBLP:conf/lics/GaoAC12,DBLP:conf/cade/GaoAC12}. Delta-decision procedures allow one-sided bounded numerical errors, which is a practically useful relaxation that significantly reduces the computational complexity of the problems. With such relaxation, SMT problems with hundreds of variables and highly nonlinear constraints (such as differential equations) have been solved in practical applications~\cite{kong2015dreach}. Existing work in this direction has focused on satisfiability of quantifier-free SMT problems. Going one level up, SMT problems with both free and universally quantified variables, which correspond to $\exists\forall$-formulas over the reals, are much more expressive. For instance, such formulas can encode the search for robust control laws in highly nonlinear dynamical systems, a central problem in robotics. Non-convex, multi-objective, and disjunctive optimization problems can all be encoded as $\exists\forall$-formulas, through the natural definition of ``finding some $x$ such that for all other $x'$, $x$ is better than $x'$ with respect to certain constraints.'' Many other examples from various areas are listed in~\cite{ratschan2012applications}. 


Counterexample-Guided Inductive Synthesis (CEGIS)~\cite{solar2008program} is a framework for program synthesis that can be applied to solve generic exists-forall problems. The idea is to break the process of solving $\exists\forall$-formulas into a loop between {\em synthesis} and {\em verification}. The synthesis procedure finds solutions to the existentially quantified variables and gives the solutions to the verifier to see if they can be validated, or falsified by {\em counterexamples}. The counterexamples are then used as learned constraints for the synthesis procedure to find new solutions. This method has been shown effective for many challenging problems, frequently generating more optimized programs than the best manual implementations~\cite{solar2008program}. However, a direct application of CEGIS to decision problems over real numbers suffers from several problems. CEGIS is complete in finite domains because it can explicitly enumerate solutions, which can not be done in continuous domains. Also, CEGIS ensures progress by avoiding duplication of solutions, while due to numerical sensitivity, precise control over real numbers is difficult. In this paper we propose methods that bypass such difficulties. 


We propose an integration of the CEGIS method in the branch-and-prune framework as a generic algorithm for solving nonlinear $\exists\forall$-formulas over real numbers and prove that the algorithm is $\delta$-complete. We achieve this goal by using CEGIS-based methods for turning universally-quantified constraints into pruning operators, which is then used in the branch-and-prune framework for the search for solutions on the existentially-quantified variables. In doing so, we take special care to ensure correct handling of numerical errors in the computation, so that $\delta$-completeness can be established for the whole procedure.

The paper is organized as follows. We first review the background, and then present the details of the main algorithm in Section 3. We then give a rigorous proof of the $\delta$-completeness of the procedure in Section 4. We demonstrated the effectiveness of the procedures on various global optimization and Lyapunov function synthesis problems in Section 5.

\paragraph{Related Work.}
Decision procedures for quantified formulas in real arithmetic have been studied in various areas such as symbolic computation and constraint solving. Cylindrical algebraic decomposition~\cite{collins}) is a well-known method for quantifier elimination for formulas that only contain polynomials. The procedures are known to have very high computational complexity (double exponential~\cite{BrownD07}), and can not handle problems with transcendental functions. Quantified constraints over real numbers have been studied in constraint programming~\cite{DBLP:conf/cp/BenhamouG00,DBLP:journals/ai/GentNRS08,DBLP:journals/tocl/Ratschan06,DBLP:journals/jsc/Ratschan02,DBLP:journals/mics/HladikR14,Nightingale2005}. In particular, the work in~\cite{DBLP:journals/tocl/Ratschan06,DBLP:journals/jsc/Ratschan02,DBLP:journals/mics/HladikR14} develops quasi-decision procedures for solving quantified constraints over the reals with a numerically-relaxed notion of completeness~\cite{DBLP:journals/jar/FranekRZ16} that is closely related to the notion of delta-completeness here. In comparison, the focus of our work (apart from improving scalability) can be seen as an extension of the same line of work that further parameterizes the procedures with explicit bounds on the numerical errors, which requires the design of various new techniques such as double-sided error control (Section 3.2). State-of-the-art SMT solvers such as CVC4~\cite{cvc4} and Z3~\cite{z3} provide limited quantified reasoning  support~\cite{Moura:2007:EES:1420853.1420872,10.1007/978-3-662-46681-0_14,GeBT-CADE-07,DBLP:conf/cav/ReynoldsDKTB15}
for decidable fragments of first-order logic and theories. Optimization Modulo Theories (OMT) is a new field that focuses on
solving a restricted form of quantified
reasoning~\cite{Nieuwenhuis2006,Cimatti:2010:SMT:2175554.2175565,Sebastiani:2012:OSL:2352896.2352936},
focusing on linear formulas. Generic
approaches for solving exists-forall problems based on the CEGIS framework have typically been used as a heuristic procedure without aiming for completeness guarantees~\cite{dutertresolving}. 
  

\section{Preliminaries}\label{sec:preliminaries}

\subsection{Delta-Decisions and $\ucnf$-Formulas}

We consider first-order formulas over real numbers that can contain arbitrary nonlinear functions that can be numerically approximated, such as polynomials, exponential, and trignometric functions. Theoretically, such functions are called {\em Type-2 computable} functions~\cite{CAbook}. We write this language as $\lrf$, formally defined as:
\begin{definition}[The $\lrf$ Language]\label{def:lrf_language}
Let $\mathcal{F}$ be the set of Type-2 computable functions. We define $\lrf$ to be the following first-order language:
\begin{align*}
t& := x \; | \; f( \vec{t} ), \mbox{ where }f\in \mathcal{F}\mbox{, possibly 0-ary (constant)};\\
\varphi& := t(\vec x)> 0 \; | \; t(\vec x)\geq 0 \; | \; \varphi\wedge\varphi \; | \; \varphi\vee\varphi \; | \; \exists x_i\varphi \; |\; \forall x_i\varphi.
\end{align*}
\end{definition}
\begin{remark}
Negations are not needed as part of the base syntax, as it can be defined through arithmetic: $\neg (t>0)$ is simply $-t\geq 0$. Similarly, an equality $t=0$ is just $t\geq 0\wedge -t\geq 0$. In this way we can put the formulas in normal forms that are easy to manipulate.
\end{remark}
We will focus on the $\exists\forall$-formulas in $\lrf$ in this paper. Decision problems for such formulas are equivalent to satisfiability of SMT with universally quantified variables, whose free variables are implicitly existentially quantified. 

It is clear that, when the quantifier-free part of an $\exists\forall$ formula is in Conjunctive Normal Form (CNF), we can always push the universal quantifiers inside each conjunct, since universal quantification commute with conjunctions. Thus the decision problem for any $\exists\forall$-formula is equivalent to the satisfiability of formulas in the following normal form:
\begin{definition}[CNF$^{\forall}$ Formulas in $\lrf$]\label{def:cnf_forall_formula}
We say an $\lrf$-formula $\varphi$ is in the CNF$^{\forall}$, if it is of the form
\begin{eqnarray}
\varphi(\vec x) := \bigwedge_{i=0}^m \Big( \forall \vec y (\bigvee_{j=0}^{k_i}c_{ij}(\vec x, \vec y)) \Big)
\end{eqnarray}
where $c_{ij}$ are atomic constraints. Each universally quantified conjunct of the formula, i.e.,
\[\forall \vec y (\bigvee_{j=0}^{k_i}c_{ij}(\vec x, \vec y))\]
is called as a {\bf $\forall$-clause}. Note that $\forall$-clauses only contain disjunctions and no nested conjunctions. If all the $\forall$-clauses are vacuous, we say $\varphi(\vec x)$ is a {\em ground SMT} formula.
\end{definition}
The algorithms described in this paper will assume that an input formula is in CNF$^{\forall}$ form. We can now define the {\em $\delta$-satisfiability} problems for $\ucnf$-formulas.
\begin{definition}[Delta-Weakening/Strengthening]\label{def:delta_weakening}
Let $\delta\in\mathbb{Q}^+$ be arbitrary. Consider an arbitrary $\ucnf$-formula of the form
\[
\varphi(\vec x) := \bigwedge_{i=0}^m \Big( \forall \vec y (\bigvee_{j=0}^{k_i} f_{ij}(\vec x, \vec y)\circ 0) \Big)
\]
where $\circ\in\{>,\geq\}$. We define the $\delta$-weakening of $\varphi(\vec x)$ to be:
\[
\varphi^{-\delta}(\vec x) := \bigwedge_{i=0}^m \Big( \forall \vec y (\bigvee_{j=0}^{k_i} f_{ij}(\vec x, \vec y)\geq -\delta) \Big).
\]
Namely, we weaken the right-hand sides of all atomic formulas from $0$ to $-\delta$. Note how the difference between strict and nonstrict inequality becomes unimportant in the $\delta$-weakening. We also define its dual, the $\delta$-strengthening of $\varphi(\vec x)$:
\[
\varphi^{+\delta}(\vec x) := \bigwedge_{i=0}^m \Big( \forall \vec y (\bigvee_{j=0}^{k_i} f_{ij}(\vec x, \vec y)\geq +\delta) \Big).
\]
\end{definition}
Since the formulas in the normal form no longer contain negations, the relaxation on the atomic formulas is implied by the original formula (and thus weaker), as was easily shown in~\cite{DBLP:conf/lics/GaoAC12}.
\begin{proposition}
For any $\varphi$ and $\delta\in\mathbb{Q}^+$, $\varphi^{-\delta}$ is logically weaker, in the sense that $\varphi\rightarrow \varphi^{-\delta}$ is always true, but not vice versa.
\end{proposition}

\begin{example}\label{example:delta-weakening}
Consider the formula
\[\forall y\; f(x,y)=0.\]
It is equivalent to the $\ucnf$-formula
\[(\forall y (-f(x,y)\geq 0)\wedge \forall y (f(x,y)\geq 0))\]
whose $\delta$-weakening is of the form
\[(\forall y (-f(x,y)\geq -\delta)\wedge \forall y (f(x,y)\geq -\delta))\]
which is logically equivalent to
\[\forall y (\norm{f(x,y)} \leq \delta).\]
We see that the weakening of $f(x,y)=0$ by $\norm{f(x,y)} \leq \delta$ defines a natural relaxation.
\end{example}
\begin{definition}[Delta-Completeness]\label{def:delta-completeness}
Let $\delta\in\mathbb{Q}^+$ be arbitrary. We say an algorithm is $\delta$-complete for $\exists\forall$-formulas in $\lrf$, if for any input $\ucnf$-formula $\varphi$, it always terminates and returns one of the following answers correctly:
\begin{itemize}
\item {\bf unsat}: $\varphi$ is unsatisfiable.
\item {\bf $\delta$-sat}: $\varphi^{-\delta}$ is satisfiable.
\end{itemize}
When the two cases overlap, it can return either answer.
\end{definition}

\subsection{The Branch-and-Prune Framework}

A practical algorithm that has been shown to be $\delta$-complete for
ground SMT formulas is the {\em branch-and-prune} method developed for
interval constraint propagation~\cite{handbookICP}. A description of
the algorithm in the simple case of an equality constraint is in
Algorithm~\ref{algo:icp}.
\begin{algorithm}[ht]
        \caption{Branch-and-Prune}\label{algo:icp}
        \begin{algorithmic}[1]
                \Function{Solve}{$f(x)=0$, $B_x$, $\delta$}
                \State $S \gets \{B_x\}$
                \While{$S \neq \emptyset$}
                        \State $B \gets S.\mathsf{pop}()$
                        \State $B' \gets\mathsf{FixedPoint}\Big(\lambda B. B\cap \prune(B,f(x)=0), B\Big)$\label{algo:icp:prune}
                        \If{$B' \not = \emptyset$}
                            \If{$\norm{f(B')} > \delta$}\label{algo:icp:delta-check}
                                \State $\{B_1,B_2\} \gets \mathsf{Branch}(B')$\label{algo:icp:branch}
                                \State $S.\mathsf{push}(\{B_1,B_2\})$
                            \Else
                                \State \Return {\bf $\delta$-sat}\label{algo:icp:delta-sat}
                            \EndIf
                        \EndIf
                \EndWhile
                \State \Return {\bf unsat}\label{algo:icp:unsat}
                \EndFunction
        \end{algorithmic}
\end{algorithm}

The procedure combines {\em pruning} and {\em branching}
operations. Let $\mathcal{B}$ be the set of all boxes (each variable
assigned to an interval), and $\mathcal{C}$ a set of constraints in
the language.  \textsf{FixedPoint(g, B)} is a procedure computing a
fixedpoint of a function $g : \mathcal{B} \to \mathcal{B}$ with an
initial input $B$.  A pruning operation
$\prune: \mathcal{B} \times \mathcal{C} \rightarrow \mathcal{B}$ takes a
box $B\in\mathcal{B}$ and a constraint as input, and returns an
ideally smaller box $B'\in \mathcal{B}$ (Line~\ref{algo:icp:prune}) that is guaranteed to
still keep all solutions for all constraints if there is any.  When
such pruning operations do not make progress, the \textsf{Branch}
procedure picks a variable, divides its interval by halves, and creates
two sub-problems $B_1$ and $B_2$ (Line~\ref{algo:icp:branch}). The procedure terminates if
either all boxes have been pruned to be empty (Line~\ref{algo:icp:unsat}), or if a small
box whose maximum width is smaller than a given threshold $\delta$ has
been found (Line~\ref{algo:icp:delta-sat}). In~\cite{DBLP:conf/cade/GaoAC12}, it has been
proved that Algorithm~\ref{algo:icp} is $\delta$-complete iff the
pruning operators satisfy certain conditions for being {\em
  well-defined} (Definition~\ref{welldefined}).


\section{Algorithm}\label{sec:algorithms}
The core idea of our algorithm for solving $\ucnf$-formulas is as follows. We view the universally quantified constraints as a special type of pruning operators, which can be used to reduce possible values for the free variables based on their consistency with the universally-quantified variables. We then use these special $\forall$-pruning operators in an overall branch-and-prune framework to solve the full formula in a $\delta$-complete way. A special technical difficulty for ensuring $\delta$-completeness is to control numerical errors in the recursive search for counterexamples, which we solve using {\em double-sided error control}. We also improve quality of counterexamples using local-optimization algorithms in the $\forall$-pruning operations, which we call {\em locally-optimized counterexamples}.

In the following sections we describe these steps in detail. For notational simplicity we will omit vector symbols and assume all variable names can directly refer to vectors of variables.

\subsection{$\forall$-Clauses as Pruning Operators}

Consider an arbitrary $\ucnf$-formula\footnote{Note that without loss of generality we only use nonstrict inequality here, since in the context of $\delta$-decisions the distinction between strict and nonstrict inequalities as not important, as explained in Definition~\ref{def:delta_weakening}.}
\[
\varphi(x) := \bigwedge_{i=0}^m \Big( \forall y (\bigvee_{j=0}^{k_i} f_{ij}(x, y)\geq 0) \Big).
\]
It is a conjunction of $\forall$-clauses as defined in Definition~\ref{def:cnf_forall_formula}. Consequently, we only need to define pruning operators for $\forall$-clauses so that they can be used in a standard branch-and-prune framework. The full algorithm for such pruning operation is described in Algorithm~\ref{algo:generic}.
\begin{algorithm}[!ht]
    \caption{$\forall$-Clause Pruning}\label{algo:generic}
    \begin{algorithmic}[1]
        \Statex
        \Function{Prune}{$B_x$, $B_y$, $\forall y \bigvee_{i=0}^k f_i(x,y)\geq 0$, $\delta'$, $\varepsilon$, $\delta$}
            \Repeat
                \State $B_x^{\mathrm{prev}} \gets B_x$
                \State $\psi \gets \bigwedge_i f_i(x,y)<0$\label{algo:generic:ce_begin}
                \State $\psi^{+\varepsilon}\gets$ \textsf{Strengthen}($\psi,\varepsilon$)\label{algo:generic:ce_strengthening}
                \State $b \gets \solve(y, \psi^{+\varepsilon}, \delta')$\label{algo:generic:find_ce} \Comment{$0 < \delta' < \varepsilon < \delta $ should hold.}
                \If{$b = \emptyset$}
                    \State \Return $B_x$ \Comment{No counterexample found, stop pruning.}\label{algo:generic:no_ce}
                \EndIf\label{algo:generic:ce_end}
                \For{$i \in \{ 0,...,k\}$}\label{algo:generic:pruning_begin}
                        \State $B_i \gets B_x \cap \prune\Big(B_x, f_{i}(x,b)\geq 0\Big)$\label{algo:generic:pruning_intersect}
                \EndFor\label{algo:generic:pruning_loop_end}
                \State $B_x \gets \bigsqcup_{i=0}^k B_{i}$\label{algo:generic:pruning_end}\label{algo:generic:pruning_hull}
            \Until{$B_{x} \neq B_x^{\mathrm{prev}}$}
            \State \Return $B_{x}$
        \EndFunction
    \end{algorithmic}
\end{algorithm}

In Algorithm~\ref{algo:generic}, the basic idea is to use special $y$ values that witness the {\em negation} of the original constraint to prune the box assignment on $x$. The two core steps are as follows.
\begin{enumerate}
\item Counterexample generation (Line~\ref{algo:generic:ce_begin} to~\ref{algo:generic:ce_end}). The query for a counterexample $\psi$ is defined as the negation of the quantifier-free part of the constraint (Line~\ref{algo:generic:ce_begin}). The method $\solve(y, \psi, \delta)$ means to obtain a solution for the variables $y$ $\delta$-satisfying the logic formula $\psi$. When such a solution is found, we have a counterexample that can falsify the $\forall$-clause on some choice of $x$. Then we use this counterexample to prune on the domain of $x$, which is currently $B_x$.

The strengthening operation on $\psi$ (Line~\ref{algo:generic:ce_strengthening}), as well as the choices of $\varepsilon$ and $\delta'$, will be explained in the next subsection. 

\item Pruning on $x$ (Line~\ref{algo:generic:pruning_begin} to~\ref{algo:generic:pruning_end}). In the counterexample generation step, we have obtained a counterexample $b$. The pruning operation then uses this value to prune on the current box domain $B_x$. Here we need to be careful about the logical operations. For each constraint, we need to take the intersection of the pruned results on the counterexample point (Line~\ref{algo:generic:pruning_intersect}). Then since the original clause contains the disjunction of all constraints, we need to take the box-hull ($\bigsqcup$) of the pruned results (Line~\ref{algo:generic:pruning_hull}).
\end{enumerate}
We can now put the pruning operators defined for all $\forall$-clauses in the overall branch-and-prune framework shown in Algorithm~\ref{algo:icp}.

The pruning algorithms are inspired by the CEGIS loop, but are different in multiple ways. First, we never explicitly compute any candidate solution for the free variables. Instead, we only prune on their domain boxes. This ensures that the size of domain box decreases (together with branching operations), and the algorithm terminates. Secondly, we do not explicitly maintain a collection of constraints. Each time the pruning operation works on previous box -- i.e., the learning is done on the model level instead of constraint level. On the other hand, being unable to maintain arbitrary Boolean combinations of constraints requires us to be more sensitive to the type of Boolean operations needed in the pruning results, which is different from the CEGIS approach that treats solvers as black boxes.

\subsection{Double-Sided Error Control}

To ensure the correctness of Algorithm~\ref{algo:generic}, it is
necessary to avoid spurious counterexamples which do \emph{not}
satisfy the negation of the quantified part in a $\forall$-clause. We
illustrate this condition by consider a \emph{wrong} derivation of
Algorithm~\ref{algo:generic} where we do not have the strengthening
operation on Line~\ref{algo:generic:ce_strengthening} and try to find
a counterexample by directly executing
$b \gets \mathrm{Solve}(y, \psi = \bigwedge_{i=0}^k f_i(x,y)< 0,
\delta)$. Note that the counterexample query $\psi$ can be highly
nonlinear in general and not included in a decidable fragment. As a
result, it must employ a delta-decision procedure (i.e. \textrm{Solve}
with $\delta' \in \mathbb{Q}^{+}$) to find a counterexample. A
consequence of relying on a delta-decision procedure in the
counterexample generation step is that we may obtain a spurious
counterexample $b$ such that for some $x = a$:
\[
  \bigwedge_{i=0}^k f_i(a,b) \le \delta
  \quad \textrm{instead of} \quad
  \bigwedge_{i=0}^k f_i(a,b) < 0.
\]
Consequently the following pruning operations fail to reduce their
input boxes because a spurious counterexample does not witness any
inconsistencies between $x$ and $y$. As a result, the fixedpoint loop
in this $\forall$-Clause pruning algorithm will be terminated
immediately after the first iteration. This makes the outer-most
branch-and-prune framework (Algorithm~\ref{algo:icp}), which employs
this pruning algorithm, solely rely on branching operations. It can
claim that the problem is $\delta$-satisfiable while providing an
arbitrary box $B$ as a model which is small enough
($\norm{B} \le \delta$) but does not include a $\delta$-solution.


To avoid spurious counterexamples, we directly strengthen the
counterexample query with $\varepsilon \in \mathbb{Q}^+$ to have
\[\psi^{+\varepsilon} := \bigwedge_{i=0}^k f_i(a,b) \leq -\varepsilon.\]
Then we choose a weakening parameter $\delta' \in \mathbb{Q}$ in
solving the strengthened formula. By analyzing the two possible
outcomes of this counterexample search, we show the constraints on
$\delta'$, $\varepsilon$, and $\delta$ which guarantee the
correctness of Algorithm~\ref{algo:generic}:
\begin{itemize}
\item {\bf $\delta'$-sat case}: We have $a$ and $b$ such that
  $\bigwedge_{i=0}^k f_i(a,b) \leq -\varepsilon + \delta'$. For
  $y = b$ to be a valid counterexample, we need
  $-\varepsilon + \delta' < 0$. That is, we have
  \begin{eqnarray}
    \delta' < \varepsilon.
  \end{eqnarray}
  In other words, the strengthening factor $\varepsilon$ should be
  greater than the weakening parameter $\delta'$ in the counterexample
  search step.
\item {\bf unsat case}: By checking the absence of counterexamples, it
  proved that
  $\forall y \bigvee_{i=0}^k f_i(x, y) \ge -\varepsilon$ for all
  $x \in B_{x}$. Recall that we want to show that
  $\forall y \bigvee_{i=0}^k f_i(x, y) \ge -\delta$ holds for some
  $x = a$ when Algorithm~\ref{algo:icp} uses this pruning algorithm
  and returns $\delta$-sat. To ensure this property, we need the
  following constraint on $\varepsilon$ and $\delta$:
  \begin{eqnarray}
    \varepsilon < \delta.
  \end{eqnarray}
\end{itemize}



\subsection{Locally-Optimized Counterexamples}
The performance of the pruning algorithm for $\ucnf$-formulas depends
on the quality of the counterexamples found during the search.

\begin{figure}[!ht]
  \centering
  \begin{subfigure}[b]{0.49\textwidth}
    \centering
    \includegraphics[width=\textwidth]{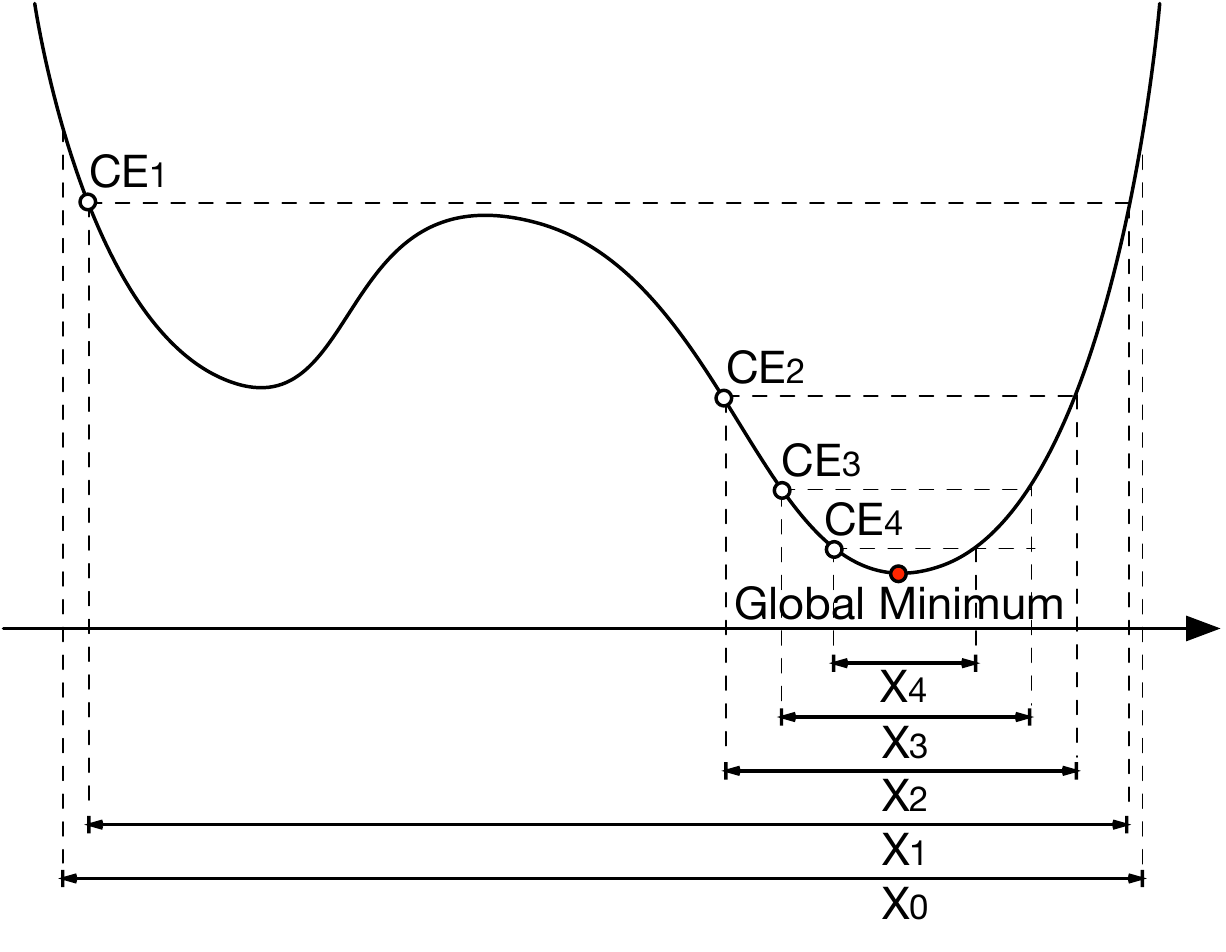}
    \caption{Without local optimization.}
    \label{fig:no_local_optimization}
  \end{subfigure}
  \begin{subfigure}[b]{0.49\textwidth}
    \centering
    \includegraphics[width=\textwidth]{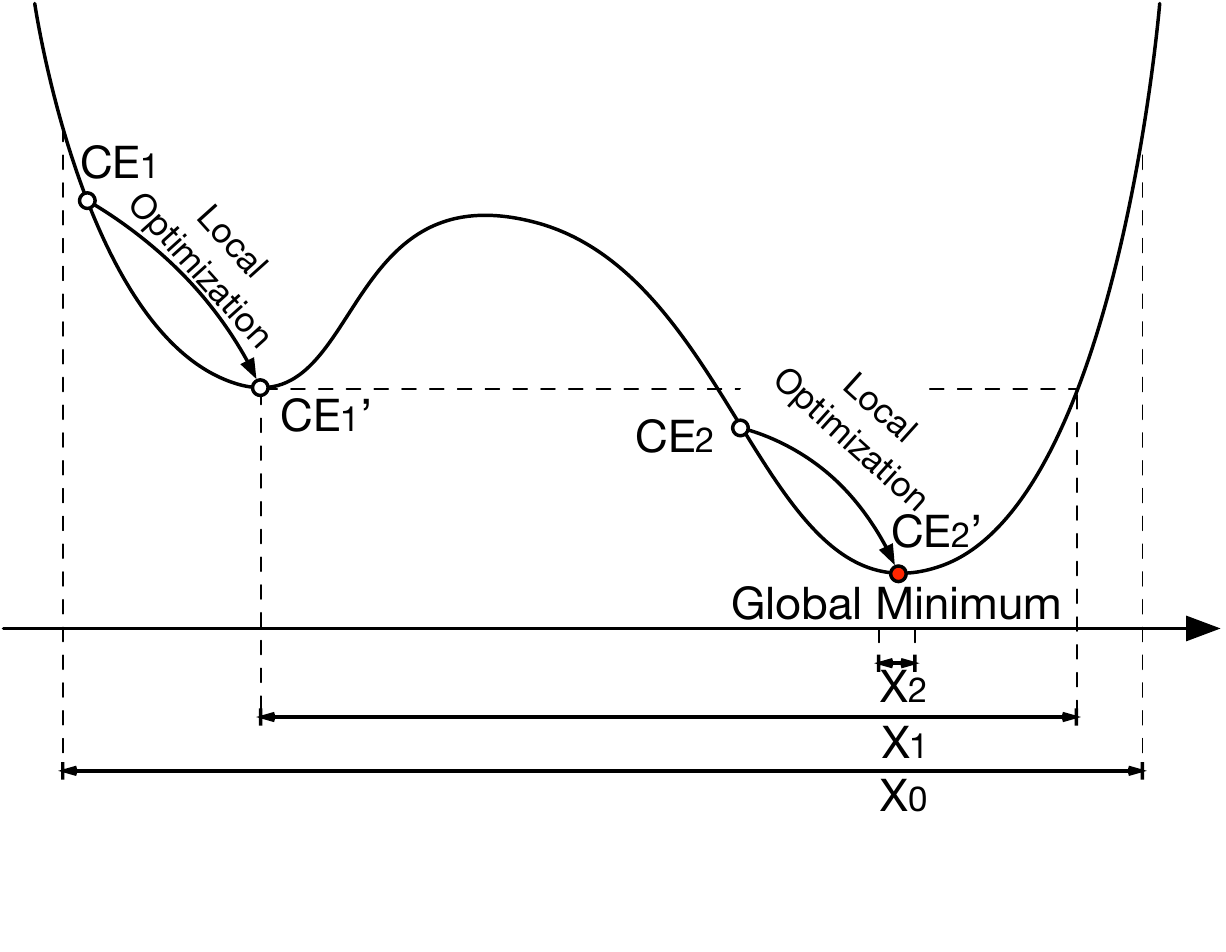}
    \caption{With local optimization.}
    \label{fig:local_optimization}
  \end{subfigure}
  \label{fig:with_and_without_local_optimization}
  \caption[]{Illustrations of the pruning algorithm for $\ucnf$-formula with and without using local optimization.}
\end{figure}

Figure~\ref{fig:no_local_optimization} illustrates this point by
visualizing a pruning process for an unconstrained minimization
problem, $\exists x \in X_0 \forall y \in X_0 f(x) \le f(y)$. As it
finds a series of counterexamples $\mathrm{CE}_1$, $\mathrm{CE}_2$,
$\mathrm{CE}_3$, and $\mathrm{CE}_4$, the pruning algorithms uses
those counterexamples to contract the interval assignment on $X$ from
$X_0$ to $X_1$, $X_2$, $X_3$, and $X_4$ in sequence. In the search for
a counterexample (Line~\ref{algo:generic:find_ce} of
Algorithm~\ref{algo:generic}), it solves the strengthened query,
$f(x) > f(y) + \delta$. Note that the query only requires a
counterexample $y = b$ to be $\delta$-away from a candidate $x$ while
it is clear that the further a counterexample is away from candidates,
the more effective the pruning algorithm is.

Based on this observation, we present a way to improve the performance
of the pruning algorithm for $\ucnf$-formulas. After we obtain a
counterexample $b$, we locally-optimize it with the counterexample
query $\psi$ so that it ``further violates'' the constraints.
Figure~\ref{fig:local_optimization} illustrates this idea. The
algorithm first finds a counterexample $\mathrm{CE}_1$ then refines it
to $\mathrm{CE}'_1$ by using a local-optimization algorithm
(similarly, $\mathrm{CE}_2 \to \mathrm{CE}'_2$). Clearly, this refined
counterexample gives a stronger pruning power than the original
one. This refinement process can also help the performance of the
algorithm by reducing the number of total iterations in the fixedpoint
loop.

The suggested method is based on the assumption that
local-optimization techniques are cheaper than finding a global
counterexample using interval propagation techniques. In our
experiments, we observed that this assumption holds practically. We
will report the details in Section~\ref{sec:evaluation}.




\section{$\delta$-Completeness}\label{sec:analysis}
We now prove that the proposed algorithm is $\delta$-complete for arbitrary $\ucnf$ formulas in $\lrf$. 
In the work of~\cite{DBLP:conf/cade/GaoAC12}, $\delta$-completeness has been proved for branch-and-prune for ground SMT problems, under the assumption that the pruning operators are {\em well-defined}. Thus, the key for our proof here is to show that the $\forall$-pruning operators satisfy the conditions of well-definedness.

The notion of a well-defined pruning operator is defined in~\cite{DBLP:conf/cade/GaoAC12} as follows.
\begin{definition}\label{welldefined}
Let $\phi$ be a constraint, and $\mathcal{B}$ be the set of all boxes in $\mathbb{R}^n$. A pruning operator is a function $\prune: \mathcal{B} \times \mathcal{C} \rightarrow \mathcal{B}$. We say such a pruning operator is well-defined, if for any $B\in \mathcal{B}$, the following conditions are true:
\begin{enumerate}
\item $\prune(B,\phi)\subseteq B$.
\item $B\cap \{a\in \mathbb{R}^n: \phi(a) \mbox{ is true.}\} \subseteq \prune(B,\phi)$.
\item Write $\prune(B,\phi) = B'$. There exists a constant $c \in \mathbb{Q}^+$, such that, if $B' \neq \emptyset$ and $\norm{B'} < \varepsilon$ for some $\varepsilon\in \mathbb{Q}^+$, then for all $a\in B'$, $\phi^{-c\varepsilon}(a)$ is true.
\end{enumerate}
\end{definition}
We will explain the intuition behind these requirements in the next proof, which aims to establish that Algorithm~\ref{algo:generic} defines a well-defined pruning operator.
\begin{lemma}[Well-definedness of $\forall$-Pruning]\label{lem:welldefined-forall-pruning}
Consider an arbitrary $\forall$-clause in the generic form
\[c(x):= \forall y\Big(f_1(x,y)\geq 0\vee ...\vee f_k(x,y)\geq 0\Big).\]
Suppose the pruning operators for $f_1\geq 0,...,f_k\geq 0$ are well-defined, then the $\forall$-pruning operation for $c(x)$ as described in Algorithm~\ref{algo:generic} is well-defined.
\end{lemma}
\begin{proof}
We prove that the pruning operator defined by Algorithm~\ref{algo:generic} satisfies the three conditions in Definition~\ref{welldefined}. Let $B_0,...,B_k$ be a sequence of boxes, where $B_0$ is the input box $B_x$ and $B_k$ is the returned box $B$, which is possibly empty.

The first condition requires that the pruning operation for $c(x)$ is reductive. That is, we want to show that $B_x \subseteq B_x^{\mathrm{prev}}$ holds in Algorithm~\ref{algo:generic}. If it does not find a counterexample (Line~\ref{algo:generic:no_ce}), we have $B_{x} = B_x^{\mathrm{prev}}$. So the condition holds trivially. Consider the case where it finds a counterexample $b$.
The pruned box $B_{x}$ is obtained through box-hull of all the $B_i$ boxes (Line~\ref{algo:generic:pruning_hull}), which are results of pruning on $B_x^{\mathrm{prev}}$ using ordinary constraints of the form $f_i(x,b)\geq 0$ (Line~\ref{algo:generic:pruning_intersect}), for a counterexample $b$. Following the assumption that the pruning operators are well-defined for each ordinary constraint $f_i$ used in the algorithm, we know that $B_i \subseteq B_x^{\mathrm{prev}}$ holds as a loop invariant for the loop from Line~\ref{algo:generic:pruning_begin} to Line~\ref{algo:generic:pruning_loop_end}. Thus, taking the box-hull of all the $B_i$, we obtain $B_{x}$ that is still a subset of $B_x^{\mathrm{prev}}$.

The second condition requires that the pruning operation does not eliminate real solutions. Again, by the assumption that the pruning operation on Line~\ref{algo:generic:pruning_intersect} does not lose any valid assignment on $x$ that makes the $\forall$-clause true. In fact, since $y$ is universally quantified, any choice of assignment $y=b$ will preserve solution on $x$ as long as the ordinary pruning operator is well-defined. Thus, this condition is easily satisfied.

The third condition is the most nontrivial to establish. It ensures that when the pruning operator does not prune a box to the empty set, then the box should not be ``way off'', and in fact, should contain points that satisfy an appropriate relaxation of the constraint. We can say this is a notion of ``faithfulness'' of the pruning operator. For constraints defined by simple continuous functions, this can be typically guaranteed by the modulus of continuity of the function (Lipschitz constants as a special case). Now, in the case of $\forall$-clause pruning, we need to prove that the faithfulness of the ordinary pruning operators that are used translates to the faithfulness of the $\forall$-clause pruning results. First of all, this condition would not hold, if we do not have the strengthening operation when searching for counterexamples (Line~\ref{algo:generic:ce_strengthening}). As is shown in Example~\ref{example:delta-weakening}, because of the weakening that $\delta$-decisions introduce when searching for a counterexample, we may obtain a {\em spurious counterexample} that does not have pruning power. In other words, if we keep using a wrong counterexample that already satisfies the condition, then we are not able to rule out wrong assignments on $x$. Now, since we have introduced $\varepsilon$-strengthening at the counterexample search, we know that $b$ obtained on Line~\ref{algo:generic:find_ce} is a true counterexample. Thus, for some $x=a$, $f_i(a,b)<0$ for every $i$. By assumption, the ordinary pruning operation using $b$ on Line~\ref{algo:generic:pruning_intersect} guarantees faithfulness. That is, suppose the pruned result $B_i$ is not empty and $\norm{B_i} \leq \varepsilon$, then there exists constant $c_i$ such that $f_i(x,b)\geq -c_i \varepsilon$ is true. Thus, we can take the $c = \min_i c_i$ as the constant for the pruning operator defined by the full clause, and conclude that the disjunction $\bigvee_{i=0}^k f_i(x,y)<-c\varepsilon$ holds for $\norm{B_x} \le \varepsilon$.
\end{proof}
Using the lemma, we follow the results in~\cite{DBLP:conf/cade/GaoAC12}, and conclude that the branch-and-prune method in Algorithm~\ref{algo:icp} is delta-complete:
\begin{theorem}[$\delta$-Completeness] For any $\delta\in\mathbb{Q}^+$, using the proposed $\forall$-pruning operators defined in Algorithm~\ref{algo:generic} in the branch-and-prune framework described in Algorithm~\ref{algo:icp} is $\delta$-complete for the class of $\ucnf$-formulas in $\lrf$, assuming that the pruning operators for all the base functions are well-defined.
\end{theorem}
\begin{proof}
Following Theorem 4.2 ($\delta$-Completeness of $\mathrm{ICP}_{\varepsilon}$) in~\cite{DBLP:conf/cade/GaoAC12}, a branch-and-prune algorithm is $\delta$-complete iff the pruning operators in the algorithm are all well-defined. Following Lemma~\ref{lem:welldefined-forall-pruning}, Algorithm~\ref{algo:generic} always defines well-defined pruning operators, assuming the pruning operators for the base functions are well-defined. Consequently, Algorithm~\ref{algo:generic} and Algorithm~\ref{algo:icp} together define a delta-complete decision procedure for $\ucnf$-problems in $\lrf$.
\end{proof}


\section{Evaluation}\label{sec:evaluation}
\paragraph{Implementation}
We implemented the algorithms on top of
dReal~\cite{DBLP:conf/cade/GaoKC13}, an open-source delta-SMT
framework. We used IBEX-lib~\cite{DBLP:conf/aaai/TrombettoniANC11} for
interval constraints pruning and
CLP~\cite{Lougee-Heimer:2003:COI:1014495.1014507} for linear
programming. For local optimization, we used
NLopt~\cite{Johnson2011}. In particular, we used SLSQP (Sequential
Least-Squares Quadratic Programming) local-optimization
algorithm~\cite{Kraft:1994:ATM:192115.192124} for differentiable
constraints and COBYLA (Constrained Optimization BY Linear
Approximations) local-optimization algorithm~\cite{powell1998direct}
for non-differentiable constraints. The prototype solver is able to
handle $\exists\forall$-formulas that involve most standard elementary
functions, including power, $\exp$, $\log$, $\sqrt{\cdot}$,
trigonometric functions ($\sin$, $\cos$, $\tan$), inverse
trigonometric functions ($\arcsin$, $\arccos$, $\arctan$), hyperbolic
functions ($\sinh$, $\cosh$, $\tanh$), etc.

\paragraph{Experiment environment}
\newcommand{\ExpCPU}{2.9 GHz Intel Core i7}
\newcommand{\ExpRAM}{16 GB}
\newcommand{\ExpOS}{MacOS 10.13.4}
All experiments were ran on a 2017 Macbook Pro with~\ExpCPU{}
and~\ExpRAM{} RAM running~\ExpOS{}. All code and benchmarks are
available at \url{https://github.com/dreal/CAV18}.

\paragraph{Parameters}
In the experiments, we chose the strengthening parameter
$\epsilon = 0.99 \delta$ and the weakening parameter in the
counterexample search $\delta' = 0.98 \delta$. In each call to NLopt,
we used $\num{1e-6}$ for both of absolute and relative tolerances on function
value, $\num{1e-3}$ seconds for a timeout, and $100$ for the maximum
number of evaluations. These values are used as stopping criteria in NLopt.

\subsection{Nonlinear Global Optimization}

\begin{table}[!t]
  \centering
  \small
  \begin{tabular}{|l|r|r|r|r|r|r|}
    \hline
    \multicolumn{1}{|c|}{\multirow{2}{*}{Name}} & \multicolumn{3}{c|}{Solution} & \multicolumn{3}{c|}{Time (sec)}      \\  \cline{2-7}
                                                & \multicolumn{1}{c|}{Global}   & \multicolumn{1}{c|}{No L-Opt.}  & \multicolumn{1}{c|}{L-Opt.} & \multicolumn{1}{c|}{No L-Opt.} & \multicolumn{1}{c|}{L-Opt.} & \multicolumn{1}{c|}{Speed Up}       \\
    \hline
    \hline
    Ackley 2D & 0.00000 & 0.00000 & 0.00000 & 0.0579 & 0.0047 & 12.32\\
    Ackley 4D & 0.00000 & 0.00005 & 0.00000 & 8.2256 & 0.1930 & 42.62\\
    Aluffi Pentini & -0.35230 & -0.35231 & -0.35239 & 0.0321 & 0.1868 & 0.17\\
    Beale & 0.00000 & 0.00003 & 0.00000 & 0.0317 & 0.0615 & 0.52\\
    Bohachevsky1 & 0.00000 & 0.00006 & 0.00000 & 0.0094 & 0.0020 & 4.70\\
    Booth & 0.00000 & 0.00006 & 0.00000 & 0.5035 & 0.0020 & 251.75\\
    Brent & 0.00000 & 0.00006 & 0.00000 & 0.0095 & 0.0017 & 5.59\\
    Bukin6 & 0.00000 & 0.00003 & 0.00003 & 0.0093 & 0.0083 & 1.12\\
    Cross in Tray & -2.06261 & -2.06254 & -2.06260 & 0.5669 & 0.1623 & 3.49\\
    Easom & -1.00000 & -1.00000 & -1.00000 & 0.0061 & 0.0030 & 2.03\\
    EggHolder & -959.64070 & -959.64030 & -959.64031 & 0.0446 & 0.0211 & 2.11\\
    Holder Table2 & -19.20850 & -19.20846 & -19.20845 & 52.9152 & 41.7004 & 1.27\\
    Levi N13 & 0.00000 & 0.00000 & 0.00000 & 0.1383 & 0.0034 & 40.68\\
    Ripple 1 & -2.20000 & -2.20000 & -2.20000 & 0.0059 & 0.0065 & 0.91\\
    Schaffer F6 & 0.00000 & 0.00004 & 0.00000 & 0.0531 & 0.0056 & 9.48\\
    Testtube Holder & -10.87230 & -10.87227 & -10.87230 & 0.0636 & 0.0035 & 18.17\\
    Trefethen & -3.30687 & -3.30681 & -3.30685 & 3.0689 & 1.4916 & 2.06\\
    W Wavy & 0.00000 & 0.00000 & 0.00000 & 0.1234 & 0.0138 & 8.94\\
    Zettl & -0.00379 & -0.00375 & -0.00379 & 0.0070 & 0.0069 & 1.01\\
    \hline
    \hline
    Rosenbrock Cubic & 0.00000 & 0.00005 & 0.00002 & 0.0045 & 0.0036 & 1.25\\
    Rosenbrock Disk & 0.00000 & 0.00002 & 0.00000 & 0.0036 & 0.0028 & 1.29\\
    Mishra Bird & -106.76454 & -106.76449 & -106.76451 & 1.8496 & 0.9122 & 2.03\\
    Townsend & -2.02399 & -2.02385 & -2.02390 & 2.6216 & 0.5817 & 4.51\\
    Simionescu & -0.07262 & -0.07199 & -0.07200 & 0.0064 & 0.0048 & 1.33\\
    \hline
  \end{tabular}
  \caption{Experimental results for nonlinear global optimization
    problems: The first 19 problems (Ackley 2D -- Zettl) are
    unconstrained optimization problems and the last five problems
    (Rosenbrock Cubic -- Simionescu) are constrained optimization
    problems. We ran our prototype solver over those instances with
    and without local-optimization option (``L-Opt.'' and ``No
    L-Opt.'' columns) and compared the results. We chose
    $\delta = 0.0001$ for all instances.}
  \label{table:experiments}
\end{table}


We encoded a range of highly nonlinear $\exists\forall$-problems from
constrained and unconstrained optimization
literature~\cite{jamil2013literature,wiki:opt_test}. Note that the
standard optimization problem
\[\min f(x) \mbox{ s.t. } \varphi(x),\ \ x\in \mathbb{R}^n,\]
 can be encoded as the logic formula:
 \[
   \varphi(x) \wedge \forall y \Big(\varphi(y)\rightarrow f(x)\leq f(y)\Big).
 \]

\begin{figure}[!t]
  \centering
  \begin{subfigure}[b]{0.475\textwidth}
    \centering
    \includegraphics[width=\textwidth]{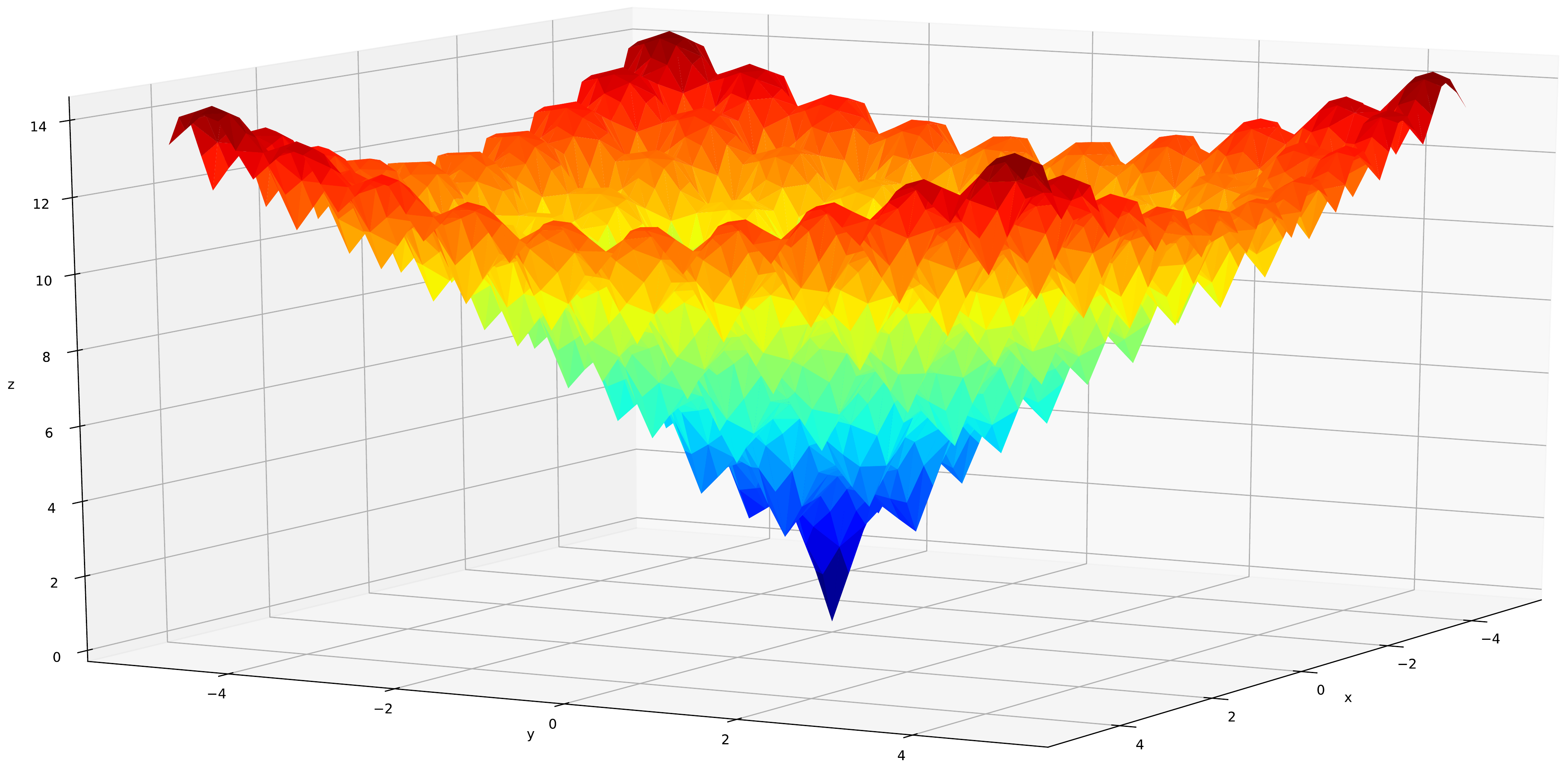}
    \caption{\small Ackley Function.}
    \label{fig:Ackley}
  \end{subfigure}
  \begin{subfigure}[b]{0.475\textwidth}
    \centering
    \includegraphics[width=\textwidth]{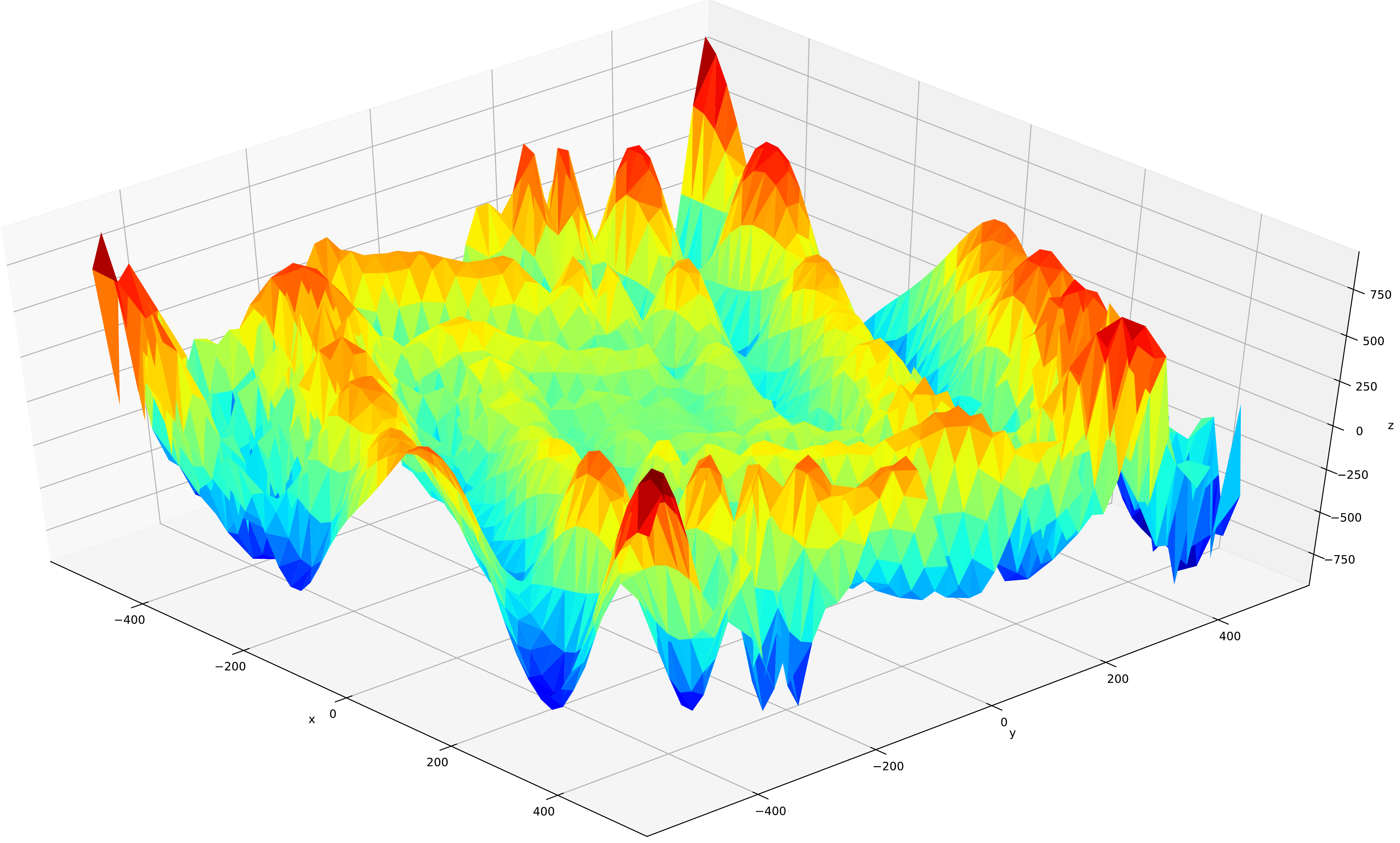}
    \caption{\small EggHolder Function.}
    \label{fig:EggHolder}
  \end{subfigure}
  \vskip\baselineskip
  \begin{subfigure}[b]{0.475\textwidth}
    \centering
    \includegraphics[width=\textwidth]{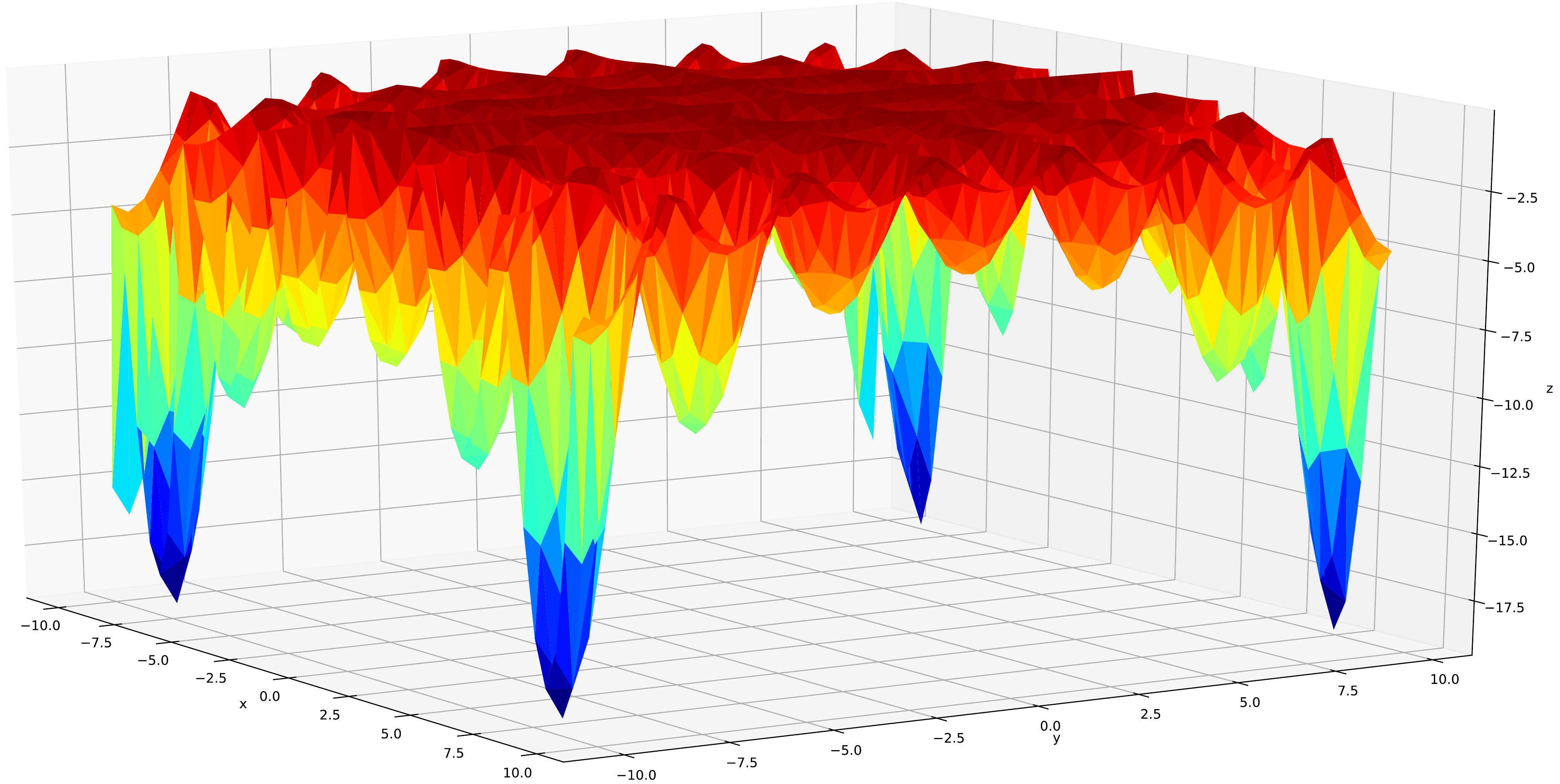}
    \caption{\small Holder Table2 Function.}
    \label{fig:Booth}
  \end{subfigure}
  \begin{subfigure}[b]{0.475\textwidth}
    \centering
    \includegraphics[width=\textwidth]{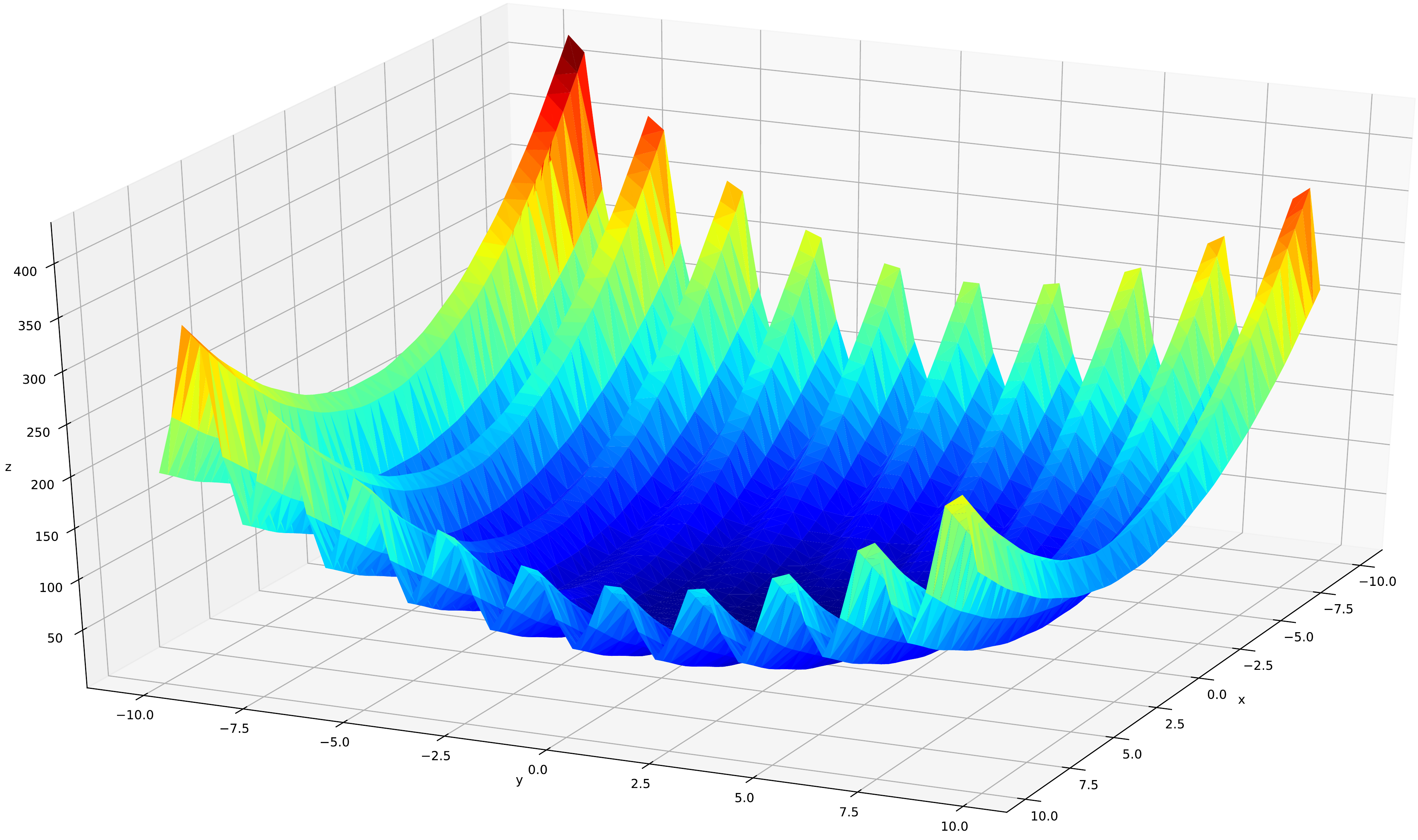}
    \caption{\small Levi N13 Function.}
    \label{fig:Booth}
  \end{subfigure}
  \begin{subfigure}[b]{0.475\textwidth}
    \centering
    \includegraphics[width=\textwidth]{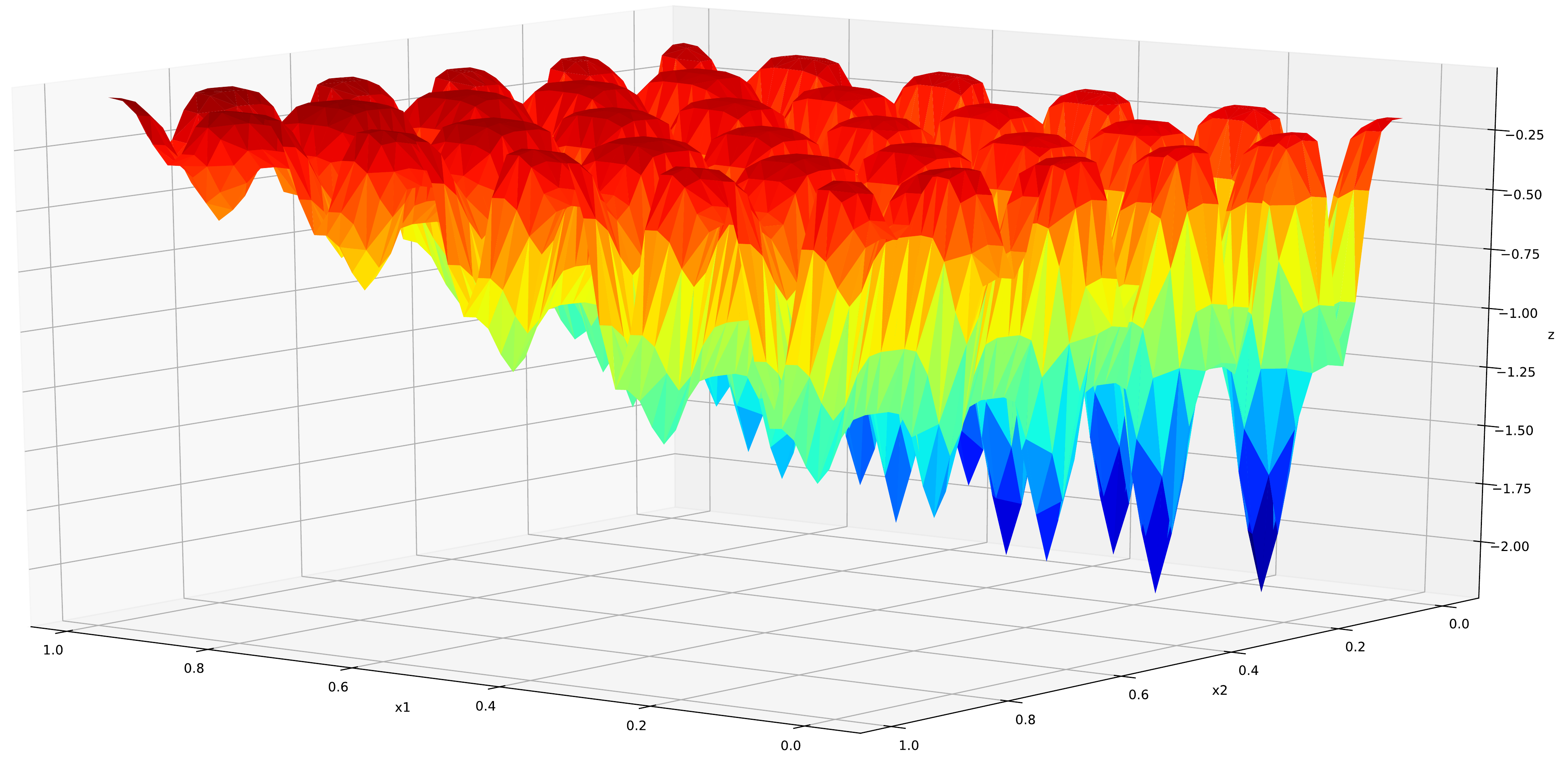}
    \caption{\small Ripple 1 Function.}
    \label{fig:Ackley}
  \end{subfigure}
  \begin{subfigure}[b]{0.475\textwidth}
    \centering
    \includegraphics[width=\textwidth]{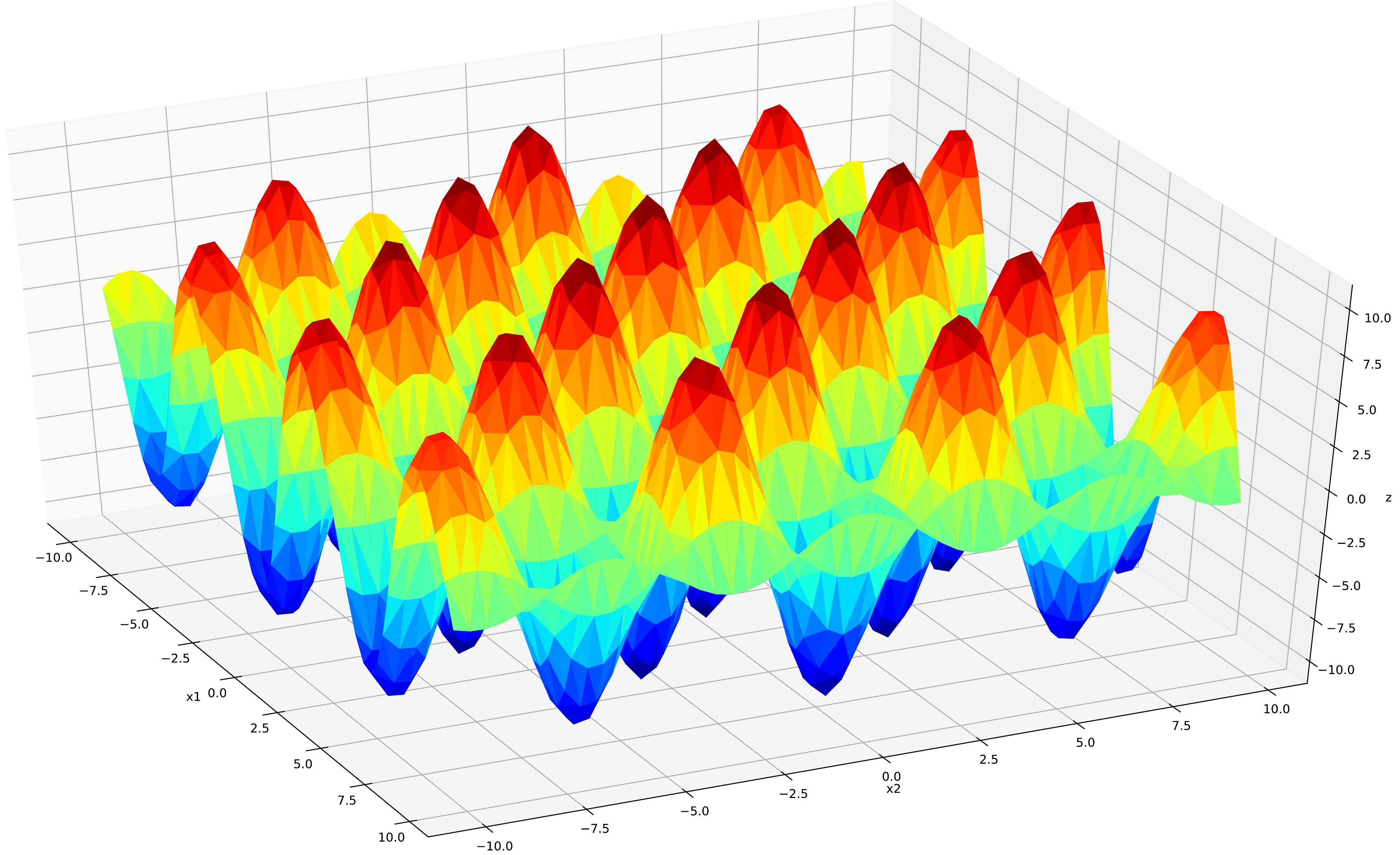}
    \caption{\small Testtube Holder Function.}
    \label{fig:TestTubeHolder}
  \end{subfigure}
  \caption[]{Nonlinear Global Optimization Examples.}
  \label{fig:nonconvexexamples}
\end{figure}


As plotted in Figure~\ref{fig:nonconvexexamples}, these optimization
problems are non-trivial: they are highly non-convex problems that are
designed to test global optimization or genetic programming
algorithms. Many such functions have a large number of local
minima. For example, Ripple 1 Function~\cite{jamil2013literature}
\[
  f(x_1, x_2) = \sum_{i=1}^2
  -e^{-2(\log2)\left(\frac{x_1 - 0.1}{0.8}\right)^2} (\sin^6(5 \pi x_i) + 0.1 \cos^2(500 \pi x_i))
\]
defined in $x_i \in [0, 1]$ has 252004 local minima with the global
minima $f(0.1, 0.1) = -2.2$. As a result, local-optimization
algorithms such as gradient-descent would not work for these problems
for itself. By encoding them as $\exists\forall$-problems, we can
perform guaranteed global optimization on these problems.

Table~\ref{table:experiments} provides a summary of the experiment
results. First, it shows that we can find minimum values which are
close to the known global solutions. Second, it shows that enabling
the local-optimization technique speeds up the solving process
significantly for 20 instances out of 23 instances.

\subsection{Synthesizing Lyapunov Function for Dynamical System}

We show that the proposed algorithm is able to synthesize Lyapunov
functions for nonlinear dynamic
systems described by a set of ODEs:
\[
\dot{\vec{x}}(t) = f_i(\vec{x}(t)), \quad \forall \vec{x}(t) \in X_i.
\]

Our approach is different from a recent
related-work~\cite{DBLP:conf/hybrid/KapinskiDSA14} where they used
dReal only to verify a candidate function which was found by a
simulation-guided algorithm. In contrast, we want to do both of search
and verify steps by solving a single $\exists\forall$-formula. Note
that to verify a Lyapunov candidate function $v : X \to \mathbb{R}^+$,
we need to show that the function $v$ satisfies the following conditions:
\begin{align*}
  \forall \vec{x} \in X\setminus\vec{0} \ v(\vec{x})(\vec{0}) & = 0\\
  \forall \vec{x} \in X \ \nabla v(\vec{x}(t))^T \cdot f_i(\vec{x}(t)) & \le 0.
\end{align*}
We assume that a Lyapunov function is a polynomial of some fixed
degrees over $\vec{x}$, that is,
$v(\vec{x}) = \vec{z}^T\vec{P}\vec{z}$ where $\vec{z}$ is a vector of
monomials over $\vec{x}$ and $P$ is a symmetric matrix. Then, we can
encode this synthesis problem into the $\exists\forall$-formula:
\begin{align*}
  \exists \vec{P}\ & [(v(\vec{x}) = (\vec{z}^T\vec{P}\vec{z})) \land \\
                   &  (\forall \vec{x} \in X\setminus\vec{0} \ v(\vec{x})(\vec{0})  = 0) \land \\
                   &  (\forall \vec{x} \in X \ \nabla v(\vec{x}(t))^T \cdot f_i(\vec{x}(t))  \le 0)]
\end{align*}

In the following sections, we show that we can handle two examples
in~\cite{DBLP:conf/hybrid/KapinskiDSA14}.

\subsubsection{Normalized Pendulum}
Given a standard pendulum system with normalized parameters
\[
  \begin{bmatrix}
    \dot{x}_1\\
    \dot{x}_2
  \end{bmatrix}
  =
  \begin{bmatrix}
    x_2\\
    -\sin(x_1) - x_2
  \end{bmatrix}
\]
and a quadratic template for a Lyapunov function
$v(\vec{x}) = \vec{x}^T\vec{P}\vec{x} = c_1x_1x_2 + c_2x_1^2 + c_3
x_2^2$, we can encode this synthesis problem into the following
$\exists\forall$-formula:
\begin{align*}
  \exists c_1c_2c_3\ \forall x_1x_2\ &
                                       [((50 c_3 x_1 x_2  + 50 x_1^2 c_1 + 50 x_2^2 c_2 > 0.5) \land \\
                                     & (100 c_1 x_1 x_2 + 50 x_2 c_3 + (- x_2 - \sin(x_1) (50 x_1 c_3 + 100 x_2 c_2)) < -0.5)) \lor \\
                                     & \neg ((0.01 \le x_1^2 + x_2^2) \land (x_1^2 + x_2^2 \le 1))]
\end{align*}

Our prototype solver takes $44.184$ seconds to synthesize the following function
as a solution to the problem for the bound $\norm{\vec{x}} \in [0.1, 1.0]$
and $c_{i} \in [0.1, 100]$ using $\delta = 0.05$:
\[
  v = 40.6843 x_1 x_2 + 35.6870 x_1^2 + 84.3906 x_2^2.
\]

\subsubsection{Damped Mathieu System} Mathieu dynamics are
time-varying and defined by the following ODEs:
\[
  \begin{bmatrix}
    \dot{x}_1\\
    \dot{x}_2
  \end{bmatrix}
  =
  \begin{bmatrix}
    x_2\\
    -x_2 - (2 + \sin(t))x_1
  \end{bmatrix}.
\]

Using a quadratic template for a Lyapunov function
$v(\vec{x}) = \vec{x}^T\vec{P}\vec{x} = c_1x_1x_2 + c_2x_1^2 + c_3
x_2^2$, we can encode this synthesis problem into the following
$\exists\forall$-formula:
\begin{align*}
  \exists c_1c_2c_3\ \forall x_1x_2t\ &
  [(50 x_1 x_2 c_2 + 50 x_1^2 c_1 + 50 x_2^2 c_3 > 0) \land \\
  & (100 c_1 x_1 x_2 + 50 x_2 c_2 + (- x_2  -   x_1(2 + \sin(t)))(50x_1 c_2 + 100 x_2 c_3) < 0) \\
&  \lor
  \neg((0.01 \le x_1^2 + x_2^2) \land
  (0.1 \le t) \land
  (t \le 1) \land
  (x_1^2 + x_2^2 \le 1))]
\end{align*}
Our prototype solver takes $26.533$ seconds to synthesize the
following function as a solution to the problem for the bound
$\norm{\vec{x}} \in [0.1, 1.0]$, $t \in [0.1, 1.0]$, and
$c_i \in [45, 98]$ using $\delta = 0.05$:
\[
   V = 54.6950 x_1 x_2 + 90.2849  x_1^2 + 50.5376 x_2^2.
\]








\section{Conclusion}\label{sec:conclusion}
We have described delta-decision procedures for solving exists-forall formulas in the first-order theory over the reals with computable real functions. These formulas can encode a wide range of hard practical problems such as general constrained optimization and nonlinear control synthesis. We use a branch-and-prune framework, and design special pruning operators for universally-quantified constraints such that the procedures can be proved to be delta-complete, where suitable control of numerical errors is crucial. We demonstrated the effectiveness of the procedures on various global optimization and Lyapunov function synthesis problems.


\bibliographystyle{splncs}
\bibliography{tau}

\begin{thebibliography}{10}

\bibitem{DBLP:conf/lics/GaoAC12}
Gao, S., Avigad, J., Clarke, E.M.:
\newblock Delta-decidability over the reals.
\newblock In: LICS. (2012)  305--314

\bibitem{DBLP:conf/cade/GaoAC12}
Gao, S., Avigad, J., Clarke, E.M.:
\newblock Delta-complete decision procedures for satisfiability over the reals.
\newblock In Gramlich, B., Miller, D., Sattler, U., eds.: IJCAR. Volume 7364 of
  Lecture Notes in Computer Science., Springer (2012)  286--300

\bibitem{kong2015dreach}
Kong, S., Gao, S., Chen, W., Clarke, E.:
\newblock {dReach}: $\delta$-reachability analysis for hybrid systems.
\newblock In: International Conference on Tools and Algorithms for the
  Construction and Analysis of Systems, Springer (2015)  200--205

\bibitem{ratschan2012applications}
Ratschan, S.:
\newblock Applications of quantified constraint solving over the
  reals-bibliography.
\newblock arXiv preprint arXiv:1205.5571 (2012)

\bibitem{solar2008program}
Solar-Lezama, A.:
\newblock Program synthesis by sketching.
\newblock University of California, Berkeley (2008)

\bibitem{collins}
Collins, G.E.:
\newblock Hauptvortrag: Quantifier elimination for real closed fields by
  cylindrical algebraic decomposition.
\newblock In: Automata Theory and Formal Languages. (1975)  134--183

\bibitem{BrownD07}
Brown, C.W., Davenport, J.H.:
\newblock The complexity of quantifier elimination and cylindrical algebraic
  decomposition.
\newblock In: ISSAC-2007

\bibitem{DBLP:conf/cp/BenhamouG00}
Benhamou, F., Goualard, F.:
\newblock Universally quantified interval constraints.
\newblock In Dechter, R., ed.: Principles and Practice of Constraint
  Programming - {CP} 2000, 6th International Conference, Singapore, September
  18-21, 2000, Proceedings. Volume 1894 of Lecture Notes in Computer Science.,
  Springer (2000)  67--82

\bibitem{DBLP:journals/ai/GentNRS08}
Gent, I.P., Nightingale, P., Rowley, A.G.D., Stergiou, K.:
\newblock Solving quantified constraint satisfaction problems.
\newblock Artif. Intell. \textbf{172}(6-7) (2008)  738--771

\bibitem{DBLP:journals/tocl/Ratschan06}
Ratschan, S.:
\newblock Efficient solving of quantified inequality constraints over the real
  numbers.
\newblock {ACM} Trans. Comput. Log. \textbf{7}(4) (2006)  723--748

\bibitem{DBLP:journals/jsc/Ratschan02}
Ratschan, S.:
\newblock Quantified constraints under perturbation.
\newblock J. Symb. Comput. \textbf{33}(4) (2002)  493--505

\bibitem{DBLP:journals/mics/HladikR14}
Hlad{\'{\i}}k, M., Ratschan, S.:
\newblock Efficient solution of a class of quantified constraints with
  quantifier prefix exists-forall.
\newblock Mathematics in Computer Science \textbf{8}(3-4) (2014)  329--340

\bibitem{Nightingale2005}
Nightingale, P.
\newblock In: Consistency for Quantified Constraint Satisfaction Problems.
  Springer Berlin Heidelberg, Berlin, Heidelberg (2005)  792--796

\bibitem{DBLP:journals/jar/FranekRZ16}
Franek, P., Ratschan, S., Zgliczynski, P.:
\newblock Quasi-decidability of a fragment of the first-order theory of real
  numbers.
\newblock J. Autom. Reasoning \textbf{57}(2) (2016)  157--185

\bibitem{cvc4}
Barrett, C., Conway, C.L., Deters, M., Hadarean, L., Jovanovi\'{c}, D., King,
  T., Reynolds, A., Tinelli, C.:
\newblock Cvc4.
\newblock In: Proceedings of the 23rd International Conference on Computer
  Aided Verification. CAV'11, Berlin, Heidelberg, Springer-Verlag (2011)
  171--177

\bibitem{z3}
De~Moura, L., Bj{\o}rner, N.:
\newblock {Z3}: An efficient {SMT} solver.
\newblock In: Proceedings of the Theory and Practice of Software, 14th
  International Conference on Tools and Algorithms for the Construction and
  Analysis of Systems. TACAS'08/ETAPS'08, Berlin, Heidelberg, Springer-Verlag
  (2008)  337--340

\bibitem{Moura:2007:EES:1420853.1420872}
Moura, L., Bj{\o}rner, N.:
\newblock Efficient e-matching for {SMT} solvers.
\newblock In: Proceedings of the 21st International Conference on Automated
  Deduction: Automated Deduction. CADE-21, Berlin, Heidelberg, Springer-Verlag
  (2007)  183--198

\bibitem{10.1007/978-3-662-46681-0_14}
Bj{\o}rner, N., Phan, A.D., Fleckenstein, L.:
\newblock $\nu$z - an optimizing {SMT} solver.
\newblock In Baier, C., Tinelli, C., eds.: Tools and Algorithms for the
  Construction and Analysis of Systems, Berlin, Heidelberg, Springer Berlin
  Heidelberg (2015)  194--199

\bibitem{GeBT-CADE-07}
Ge, Y., Barrett, C., Tinelli, C.:
\newblock Solving quantified verification conditions using satisfiability
  modulo theories.
\newblock In Pfenning, F., ed.: Proceedings of the 21st International
  Conference on Automated Deduction (CADE-21), Bremen, Germany. Volume 4603 of
  Lecture Notes in Computer Science., Springer (2007)  167--182

\bibitem{DBLP:conf/cav/ReynoldsDKTB15}
Reynolds, A., Deters, M., Kuncak, V., Tinelli, C., Barrett, C.W.:
\newblock Counterexample-guided quantifier instantiation for synthesis in
  {SMT}.
\newblock In Kroening, D., Pasareanu, C.S., eds.: Computer Aided Verification -
  27th International Conference, {CAV} 2015, San Francisco, CA, USA, July
  18-24, 2015, Proceedings, Part {II}. Volume 9207 of Lecture Notes in Computer
  Science., Springer (2015)  198--216

\bibitem{Nieuwenhuis2006}
Nieuwenhuis, R., Oliveras, A.
\newblock In: On SAT Modulo Theories and Optimization Problems. Springer Berlin
  Heidelberg, Berlin, Heidelberg (2006)  156--169

\bibitem{Cimatti:2010:SMT:2175554.2175565}
Cimatti, A., Franz{\'e}n, A., Griggio, A., Sebastiani, R., Stenico, C.:
\newblock Satisfiability modulo the theory of costs: Foundations and
  applications.
\newblock In: Proceedings of the 16th International Conference on Tools and
  Algorithms for the Construction and Analysis of Systems. TACAS'10, Berlin,
  Heidelberg, Springer-Verlag (2010)  99--113

\bibitem{Sebastiani:2012:OSL:2352896.2352936}
Sebastiani, R., Tomasi, S.:
\newblock Optimization in {SMT} with {LA(Q)} cost functions.
\newblock In: Proceedings of the 6th International Joint Conference on
  Automated Reasoning. IJCAR'12, Berlin, Heidelberg, Springer-Verlag (2012)
  484--498

\bibitem{dutertresolving}
Dutertre, B.:
\newblock Solving exists/forall problems with yices.
\newblock In: Workshop on Satisfiability Modulo Theories. (2015)

\bibitem{CAbook}
Weihrauch, K.:
\newblock Computable Analysis: An Introduction.
\newblock Springer-Verlag New York, Inc., Secaucus, NJ, USA (2000)

\bibitem{handbookICP}
Benhamou, F., Granvilliers, L.:
\newblock Continuous and interval constraints.
\newblock In Rossi, F., van Beek, P., Walsh, T., eds.: Handbook of Constraint
  Programming.
\newblock Elsevier (2006)

\bibitem{DBLP:conf/cade/GaoKC13}
Gao, S., Kong, S., Clarke, E.M.:
\newblock {dReal}: An {SMT} solver for nonlinear theories over the reals.
\newblock In: CADE. (2013)  208--214

\bibitem{DBLP:conf/aaai/TrombettoniANC11}
Trombettoni, G., Araya, I., Neveu, B., Chabert, G.:
\newblock Inner regions and interval linearizations for global optimization.
\newblock In Burgard, W., Roth, D., eds.: Proceedings of the Twenty-Fifth
  {AAAI} Conference on Artificial Intelligence, {AAAI} 2011, San Francisco,
  California, USA, August 7-11, 2011, {AAAI} Press (2011)

\bibitem{Lougee-Heimer:2003:COI:1014495.1014507}
Lougee-Heimer, R.:
\newblock The common optimization interface for operations research: Promoting
  open-source software in the operations research community.
\newblock IBM J. Res. Dev. \textbf{47}(1) (January 2003)  57--66

\bibitem{Johnson2011}
Johnson, S.G.:
\newblock The NLopt nonlinear-optimization package. (2011)

\bibitem{Kraft:1994:ATM:192115.192124}
Kraft, D.:
\newblock Algorithm 733: Tomp--fortran modules for optimal control
  calculations.
\newblock ACM Trans. Math. Softw. \textbf{20}(3) (September 1994)  262--281

\bibitem{powell1998direct}
Powell, M.:
\newblock Direct search algorithms for optimization calculations.
\newblock Acta numerica \textbf{7} (1998)  287--336

\bibitem{jamil2013literature}
Jamil, M., Yang, X.S.:
\newblock A literature survey of benchmark functions for global optimisation
  problems.
\newblock International Journal of Mathematical Modelling and Numerical
  Optimisation \textbf{4}(2) (2013)  150--194

\bibitem{wiki:opt_test}
{Wikipedia contributors}:
\newblock Test functions for optimization --- {Wikipedia}{,} {The Free
  Encyclopedia} (2017)

\bibitem{DBLP:conf/hybrid/KapinskiDSA14}
Kapinski, J., Deshmukh, J.V., Sankaranarayanan, S., Arechiga, N.:
\newblock Simulation-guided lyapunov analysis for hybrid dynamical systems.
\newblock In: HSCC'14, Berlin, Germany, April 15-17, 2014. (2014)  133--142

\end{thebibliography}

\end{document}